\documentclass[aps,pra,twocolumn,groupedaddress, fleqn]{revtex4}

\usepackage[latin1]{inputenc}
\usepackage[T1]{fontenc}
\usepackage[english]{babel}
\usepackage[dvips]{graphicx}
\usepackage{enumerate,amsthm,amsmath,amssymb,color,graphicx,bbm,float,framed}
\usepackage{natbib}

\newcommand{\F}{\mathbb{F}}
\newcommand{\E}{\mathbb{E}}

\newcommand{\IP}[2]{IP(#1 || #2)}
\newcommand{\Z}[1]{\IP{a_{#1}}{b_{#1}}}

\newtheorem{theorem}{Theorem}

\makeatletter
\newtheorem*{rep@theorem}{\rep@title}
\newcommand{\newreptheorem}[2]{%
\newenvironment{rep#1}[1]{%
 \def\rep@title{#2 \ref{##1}}%
 \begin{rep@theorem}}%
 {\end{rep@theorem}}}
\makeatother

\newtheorem{theo}{Theorem} 
\newtheorem{lemma}[theo]{Lemma}
\newtheorem{proposition}{Proposition}

\newreptheorem{lemma}{Lemma}

\usepackage{epstopdf,hyperref}

\newcommand{\be}{\begin{equation}}
\newcommand{\ee}{\end{equation}}
\newcommand{\etal}{\textit{et al.} }
\newcommand{\eps}{\varepsilon}
\newtheorem{definition}{Definition}

\usepackage{mathtools}
\def\multiset#1#2{\ensuremath{\left(\kern-.3em\left(\genfrac{}{}{0pt}{}{#1}{#2}\right)\kern-.3em\right)}}

\makeatletter

\renewenvironment{widetext@grid}{%
  \par\ignorespaces
  \setbox\widetext@top\vbox{%
   \vskip15\p@
   \hb@xt@\hsize{%
    \leaders\hrule\hfil
    \vrule\@height6\p@
   }%
   \vskip6\p@
  }%
  \setbox\widetext@bot\hb@xt@\hsize{%
    \vrule\@depth6\p@
    \leaders\hrule\hfil
  }%
  \onecolumngrid
  \let\set@footnotewidth\set@footnotewidth@ii
}{%
  \par
  \twocolumngrid\global\@ignoretrue
  \@endpetrue
}%

\makeatother

\begin{document}

\title{Arbitrarily long relativistic bit commitment}

\author{Kaushik Chakraborty, Andr{\'e} Chailloux, Anthony Leverrier}
\affiliation{Inria, EPI SECRET, B.P. 105, 78153 Le Chesnay Cedex, France}


\begin{abstract}
We consider the recent relativistic bit commitment protocol introduced by Lunghi \textit{et al} [\textit{Phys.~Rev.~Lett.}~2015] and present a new security analysis against classical attacks. In particular, while the initial complexity of the protocol scaled double-exponentially with the commitment time, our analysis shows that the correct dependence is only linear. 
This has dramatic implications in terms of implementation: in particular, the commitment time can easily be made arbitrarily long, by only requiring both parties to communicate classically and perform efficient classical computation.
\end{abstract}

\maketitle

Over the last decades, which witnessed the rapid expansion of quantum information, a new trend has developed: trying to obtain security guarantees based solely on the laws of physics. Perhaps the most compelling example is quantum key distribution \cite{BB84,SBC09} where two distant parties can exploit quantum theory to extract unconditionally secure keys provided that they have access to an untrusted quantum channel and an authenticated classical channel.
However, many cryptographic applications cannot be obtained only with secure key distribution. One important example is two-party cryptography, which deals with the setting where Alice and Bob want to perform a cryptographic task but do not trust each other. This is in contrast with key distribution where Alice and Bob cooperate and fight against a possible eavesdropper. 

Two-party cryptography has numerous applications, ranging from authentication to distributed cryptography in the cloud. These protocols are usually separated into building blocks, called \emph{primitives}. One of the most studied primitives is \emph{bit commitment}, which often gives a strong indication of whether two-party cryptography is possible or not in a given model. 
For example, there are many constructions of bit commitment protocols under computational assumptions \cite{BrassardChCr88, Naor91, HM96, Halevi99}.  It is then natural to ask whether quantum theory can provide security for two-party cryptographic primitives such as bit commitment or oblivious transfer. A general no-go theorem was proved in 1996 by Mayers and Lo-Chau \cite{LC97, May97}. Several attempts were made to circumvent this impossibility result by limiting the storage possibilities of the cheating party \cite{DFSS08,WST08}.
An alternative approach to obtain secure primitives, pioneered by Kent \cite{Kent99}, consists in combining quantum theory with special relativity, more precisely with the physical principle that information cannot propagate faster than the speed of light. 
This has opened the way to new, secure, bit commitment protocols  \cite{Kent11,CK12,Kent12,KTH13}, with the caveat that the commitment time is not arbitrary long in general but depends on the physical distance between the parties or on the number of parties involved. 

A major open question of the field is therefore to design a secure practical bit commitment protocol, for which the commitment time can be increased arbitrarily at a reasonable cost in terms of implementation complexity. 
In this paper, we examine a protocol due to Lunghi \etal \cite{LKB14}, which is itself adapted from based on an earlier proposal of Simard \cite{sim07}. In their recent breakthrough paper, Lunghi \etal showed that it was possible to extend the commitment time by using a multi round generalization of the Simard protocol, and established its security against classical adversaries. Unfortunately, the required resources scale double exponentially with the commitment time, making the protocol impractical for realistic applications. For instance, with the optimal configuration on Earth (meaning that each party has agents occupying antipodal locations on Earth), the commitment time is limited to less than a second. Here, we provide a new security analysis establishing that the dependence is in fact linear, provided that the dishonest player is classical. This implies that arbitrary long commitment times can be achieved even if both parties are only a few kilometers apart. We first present the relativistic bit commitment scheme studied by Lunghi \etal and we will then establish its security. 

{\bf The Lunghi \etal protocol}.--- \label{prot}
We first recall the protocol as well as the security definitions used and timing constraints.  
Both players, Alice and Bob, have agents $\mathcal{A}_1, \mathcal{A}_2$ and $\mathcal{B}_1, \mathcal{B}_2$ present at two spatial locations 1 and 2. 
Let us consider the case where Alice makes the commitment. 
The protocol (followed by honest players) consists of 4 phases: preparation, commit, sustain and reveal. The sustain phase is itself composed of many rounds, and each such round involves a pair of agents (alternating between locations 1 and 2) referred to as the active players. Overall the bit commitment protocol goes as follows. 
\begin{enumerate}
	\item \emph{Preparation phase}: $\mathcal{A}_1,\mathcal{A}_2$ (resp.~$\mathcal{B}_1,\mathcal{B}_2$) share $k$ random numbers $a_1,\dots,a_{k}$ (resp.~$b_1,\dots,b_{k}$) $\in \F_q$, for even $k$. Here, $q$ is a prime power $p^n$ for some prime $p$ and $\F_q$ refers to the Galois field of order $q$. 
	\item \emph{Commit phase}: $\mathcal{B}_1$ sends $b_1$ to $\mathcal{A}_1$, who returns $y_1 = a_1 + (d * b_1)$ where $d \in \{0,1\}$ is the committed bit. 
	\item \emph{Sustain phase}: at round $i$, active Bob sends $b_i \in \F_q$ to active Alice, who returns $y_i = a_i + (a_{i-1} * b_i)$.
	\item \emph{Reveal phase}: $\mathcal{A}_1$ reveals $d$ and $a_k$ to $\mathcal{B}_1$. $\mathcal{B}_1$ checks that $a_k = y_k + (a_{k-1} * b_k)$.
\end{enumerate}
Here, $+$ and $*$ refer to the field addition and multiplication in $\F_q$.

{\bf Security definition}.--- We follow the definitions of Ref.~\cite{LKB14}. The security requirements differ in the case of honest Alice and honest Bob. In the former case, Bob should not be able to guess the committed value right before the reveal phase. The protocol should therefore be \textit{hiding}, and it will actually be perfectly hiding here, meaning that Bob cannot guess the committed bit value better than with a random guess. 
Security for honest Bob is defined differently: the protocol should be \textit{binding}, meaning that Alice should not be able to decide the value of the committed bit after the commit phase. We follow the standard definition for bit commitment (also used in \cite{LKB14}). Let $p_d$ the probability that the Alice successfully reveals bit value $d$. We say that the protocol is $\eps$-binding if $p_0 + p_1 \le 1 + \eps$.

{\bf Timing constraints for the protocol}.---
The two pairs $(\mathcal{A}_1,\mathcal{B}_1)$ and $(\mathcal{A}_2,\mathcal{B}_2)$ are at a certain distance $d$. At each round $j$, there is an \emph{active} (Alice, Bob) pair that performs the protocol while the other, \emph{passive}, pair waits. 
At the end of round $j$, they switch roles and perform round $j+1$. 

We require that round $j$ finishes before any information about $b_{j-1}$ reaches the other Alice. For any $j$, we therefore have the following : active Alice has no information about $b_{j-1}$. This means that $y_j$ is independent of $b_{j-1}$. This will be crucial in order to show security of the protocol. 
	
\includegraphics[width = 8cm]{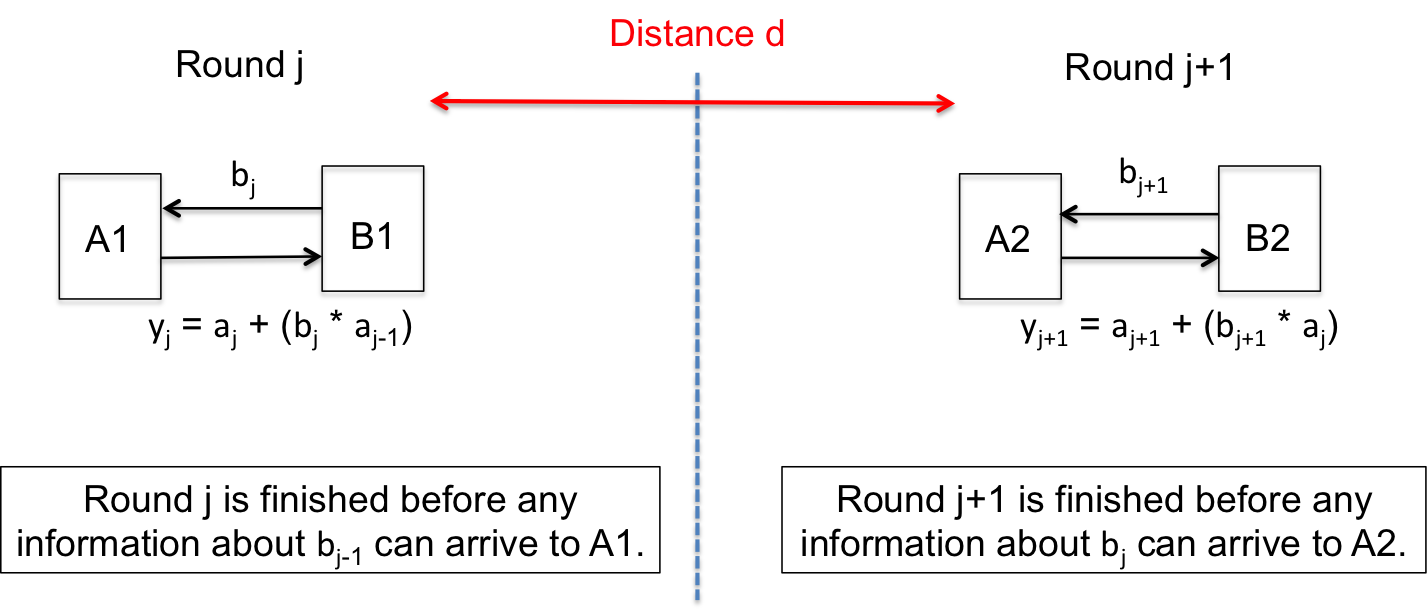}

{\bf Our result}.---
Our main contribution is to present an improved security proof for this protocol. In particular, this allows for implementations of this protocol that last for an (almost) arbitrary amount of time while the previous implementations were only secure for (much less than) a second \cite{LKB14}.

In order to prove the security of the protocol, we present an inductive argument on the number of rounds of the protocol and show that at each round, the cheating parameter for Alice increases by at most $2^{-(N-1)/2}$, where $N$ is the number of transmitted bits per round. Interestingly, the proof involves the study of $\mathrm{CHSH}_q$, which is a generalization of the $\mathrm{CHSH}$ game in the field $\F_q$.
Lunghi \etal  also studied an extension of the $\mathrm{CHSH}_q$ game, which they called ``Number on the Forehead game''. However, their security proof quickly becomes inefficient as the number of rounds increases.

{\bf The $\mathbf{CHSH_q}$ game}.---
A crucial tool of our security proof is the analysis of the CHSH$_{q}$ game introduced by Buhrman and Massar \cite{BM05}. 
This game is a natural generalisation of the CHSH game to the field $\F_q$, where two non-communicating parties, Alice and Bob, are each given an input $x$ and $y$ chosen uniformly at random from $\F_q$, and must output two numbers $a, b \in \F_q$. They win the game whenever the condition $a + b = x *y$ is satisfied. The
$\mathrm{CHSH}_q$ game has been much less studied in the litterature \cite{BM05, LLP10, LKB14} than its $q=2$ variant (see \cite{BCP14} for a recent review on nonlocality). A recent result by Bravarian and Shor \cite{BS15} establishes rather tight bounds on the classical and quantum values of the $\mathrm{CHSH}_q$ game. In particular, for prime or odd power of prime $q$ , one has:
$$\omega(\mathrm{CHSH}_q) =O(q^{-1/2-\eps_0}), \:
\omega^*(\mathrm{CHSH}_q) \leq  \frac{q-1}{q}\frac{1}{\sqrt{q}} + \frac{1}{q},
$$
for some absolute constant $\eps_0>0$.

These results hold only for a uniform input distribution. In order to use our inductive technique, we need to bound the value of this game for unbalanced inputs. It appears that the result of Bavarian and Shor doesn't easily extend to this setting. We therefore developed new proof techniques that are based on using non-signaling constraints for the study of classical strategies.

Let us consider a family of games, denoted by $\mathrm{CHSH}_q(p)$, where games are parametrized by the probability distribution $\{p_x\}_{x \in \F_q}$ for Alice's input $x$ satisfying the constraint $\max_x p_x \leq p$. For these games, Bob's input distribution is uniform over $\F_q$. In particular, $\mathrm{CHSH}_q(1/q) = \{ \mathrm{CHSH}_q \}$.
The special case with $q=2$ was considered in \cite{LLP10} where the following results are proved:
\begin{align*}
&\omega(\mathrm{CHSH}_2(p)) = (1+p)/2, \\
&\omega^*(\mathrm{CHSH}_2(p)) \leq(1+ \sqrt{p^2 +(1-p)^2})/2.
\end{align*}
Note that for $q=2$, Alice's input distribution is entirely determined by the value of $p$.
In order to prove upper bounds on the value of games in $\mathrm{CHSH}_q(p)$, we show that if Alice and Bob can win such a game with high probability then Alice has a method to obtain some information about Bob's input, something that is prohibited by the non-signaling principle. This technique doesn't directly extend to the quantum setting because Alice's method requires her to perform her game strategy for different inputs, which could disturb the underlying shared entangled state. 

Our main technical result is an upper bound on the classical value for games in $\mathrm{CHSH}_q(p)$.
 
\begin{lemma}
\label{lemma}
	For any game $G \in \mathrm{CHSH}_q(p)$, we have 
	\begin{align}
	\omega(G) \le p + \sqrt{\frac{2}{q}}.
	\end{align}
\end{lemma}
\begin{proof}
Fix a game $G \in \mathrm{CHSH}_q(p)$. As usual, the classical value of the game can always be achieved with a deterministic strategy, meaning that without loss of generality, Alice and Bob's strategies can be modeled by functions $f$ and $g$, namely: $a = f(x)$ and $b=g(y)$. 
Define the variable $r_x^{y}$ equal to $1$ if $f(x) + g(y) = x *y$ and $0$ otherwise. 

Our proof is by contradiction: if $\omega(G)$ is too large, then Alice could use her box to obtain some information about $y$, which is prohibited by non signaling.
More precisely, consider the following strategy for Alice: pick a random pair of distinct inputs $x, x'$ according to the distribution $\{p\}_{x \in \F_q}$, i.e. with probability $p_x p_x'/D$ where $D = \sum_{x \ne x'} p_x p_x'$, and output the guess $\hat{y}$ for $y$ defined by $\hat{y} = (f(x)-f(x'))*(x-x')^{-1}$.
Denote by $S_y$ the probability of correctly guessing the value $y$. Non signaling imposes that $\mathbbm{E}_y[S_y] = 1/q$, since the value $y$ is uniformly distributed in $\F_q$.

On the other hand, we note that if the game $G$ is won for both inputs $(x,y)$ and $(x',y)$, then Alice's strategy outputs the correct value for $y$. Indeed, winning the game implies that $f(x)-f(x') = (x-x')*y$ and therefore $\hat{y} = y$.
One immediately obtains a lower bound on $S_y$:
$$S_y \geq \frac{1}{D} \sum_{x \ne x'} p_x r_x^y p_x' r_{x'}^y \geq  \sum_{x \ne x'} p_x r_x^y p_x' r_{x'}^y .$$
Consider the quantity $\omega^y = \sum_x p_x r_x^y$. It satisfies:
$$(\omega^y)^2 \leq \sum_x p_x^2 (r_x^y)^2 + 2 S_y = \sum_x (p_x)^2 r_x^y + 2 S_y \leq p \omega^y + 2s_y,$$
where we used that $(p_x)^2 \leq \left(\max_{x} \{p_x\} \right) p_x \leq p p_x$. 
This implies that 
$$ \omega^y \leq \frac{1}{2}\left( p  +\sqrt{p^2 + 8S_y}\right) \leq  p + \sqrt{2S_y},$$
where the last inequality results from the concavity of the square-root function.

Finally, $\omega(G) = \mathbbm{E}_y[\omega^y]$ by definition, and therefore:
$$\omega(G) \leq p + 2\mathbbm{E}_y [\sqrt{S_y}] \leq p + \sqrt{2} \sqrt{\mathbbm{E}_y[S_y]} \leq p +\sqrt{2/q},$$
which concludes the proof.
\end{proof}

{\bf  Security of the protocol}.--- 
The perfect hiding property of this protocol has already been discussed in \cite{LKB14}. Indeed, at any point before the reveal phase, the Bobs have no information about the committed bit $d$. Our main contribution is the following binding property of this protocol.

\begin{theorem}
\label{theorem}
	This relativistic bit commitment scheme is $\eps$-binding with $\eps \le 2k\sqrt{\frac{2}{q}}$ where $k$ is the number of rounds used in the protocol. 
\end{theorem}

\begin{proof}
We present here the main elements of the proof. The technical details can be found in the Appendix. Let us fix a cheating strategy for Alice, which consists of the messages $y_j$ that het agents will send depending on the current history and the bit $d$ she wants to decommit to. During the reveal phase, Alice successfully reveals $d$ if $\mathcal{A}_1$ sends the correct $a_k$ to Bob. For a fixed cheating strategy, $a_k$ is a function of $d,b_1,\dots,b_k$. However, during the reveal phase, $\mathcal{A}_1$ has no information about $b_k$. Therefore, $\mathcal{A}_1$ will not be able to reveal $a_k$ if it has too much dependence in $b_k$ on average on $d$ . We show that this is indeed the case. 

Let $P^d_j$ the maximal probability that the passive players guesses $a_j$, given $d$. We have by definition $$
P^0_k + P^1_k = 1 + \eps.$$
In order to prove our statement, we show the following:
\begin{itemize}
	\item $P^0_1 + P^1_1 \le 1 + 2\sqrt{\frac{2}{q}}$.
	\item For any $d$ and $j$, $P^d_j \le P^d_{j-1} + \sqrt{\frac{2}{q}}$.
\end{itemize}

To prove the first point, the idea is to reduce $\mathcal{A}_2$'s strategy for guessing $a_1$ into a strategy for $\mathrm{CHSH}_q(1/2)$. $\mathcal{A}_1$ receives $b_1$ and outputs $y_1$ which is independent of $d$. $\mathcal{A}_2$ knows $d$ and outputs $a_1$. $\mathcal{A}_2$ outputs the correct $a_1$ when $a_1 + y_1 = d * b_1$. For an average $d$, this can happen with probability at most $\mathrm{CHSH}_q(1/2) \le \frac{1}{2} + \sqrt{\frac{2}{q}}$. Therefore, we have 
$$
\frac{1}{2}\left(P^0_1 + P^1_1\right) \le \mathrm{CHSH}_q(1/2) \le \frac{1}{2} + \sqrt{\frac{2}{q}}
$$
which gives the desired result. The idea here is to reduce passive Alice's strategy for guessing $a_1$ to a strategy for winning $\mathrm{CHSH}_q(1/2)$.

Similarly, fix a round $j$ and $d$. We can reduce passive Alice's strategy for guessing $a_j$ to a strategy for winning $\mathrm{CHSH}_q(P^d_{j-1})$. Indeed, active Alice knows $b_{j}$ and outputs $y_j$. Passive Alice knows $a_{j-1}$ and outputs a guess $a_j$. She outputs the correct value if and only if $a_j + y_j = b_j * a_{j-1}$.

This corresponds to an instance of $\mathrm{CHSH}_q$ where $b_j \in \F_q$ is random and where active Alice (we consider here active Alice at round $j$, which is the passive Alice at round $j-1$) can guess $a_{j-1}$ with probability $P^d_{j-1}$. This means that we can reduce passive Alice's strategy for guessing $a_j$ to a strategy for winning a certain game in $\mathrm{CHSH}_q(P^d_{j-1})$. Using Proposition \ref{Proposition:CHSHq}, we obtain $P^d_{j} \le P^d_{j-1} + \sqrt{\frac{2}{q}}$. Putting all this together, we can conclude that $P^0_k + P^1_k = 1 + 2k \sqrt{\frac 2 q}$.
\end{proof}

{\bf Experimental perspectives and open questions}.--- Let us discuss the security of the protocol in realistic conditions. Theorem \ref{theorem} shows that $m = \eps \sqrt{q/2}$ rounds can be performed for a given level of security $\epsilon$. 
In particular, if the distance between $\mathcal{A}_1/\mathcal{B}_1$ and $\mathcal{A}_2/\mathcal{B}_2$ is $d$, then the commitment can be sustained for a time $$T =  (d/c) \ \eps  \sqrt{q/2},$$
where $c$ is the speed of light.
In particular, provided that $q \gg 1/\eps^2$, the commitment time can be made arbitrary long. 
For instance, taking 128 bits of security, i.e. $\eps = 2^{-128}$ and $q=2^{340}$ gives $T \approx 3\cdot 10^{12} (d/c)$, that is approximately 30 years for a distance $d =100$ km. In this example, the messages sent at each round only consist of 340 bits.

It is also possible to reduce the distance between $\mathcal{A}_1/\mathcal{B}_1$ and $\mathcal{A}_2/\mathcal{B}_2$, at the condition that both the computation time and the communication time between $\mathcal{A}_i$ and $\mathcal{B}_i$ remains negligible compared to $d/c$. This is necessary to enforce the non-signaling condition of the $\mathrm{CHSH}_q$ game. 
For instance, if the computation time is on the order of the microsecond, then $d$ should be at least 300 meters.

Let us conclude by mentioning a few open questions. Certainly the most pressing one concerns the security of the protocol against quantum adversaries. A first step in that direction would be to obtain tight upper bounds on the entangled value $\omega^*$ of games in $\mathrm{CHSH}_q(p)$. 
Another outstanding problem is whether the bit-commitment protocol of \cite{LKB14} can be used to obtain an protocol for Oblivious-Transfer \cite{Kil88}. In particular, this would pave the way for arbitrary two-party cryptography with security based on the non-signaling principle. 
Finally, it would be particularly interesting to understand whether 2 agents are indeed necessary for each player, or whether the second agent could for instance be replaced by assuming that the spatial positions of Alice and Bob are known.

{\bf Note added}.--- In an independent and concurrent work, Fehr and Fillinger \cite{FF15} proved a general composition theorem for two-prover commitments which implies a similar bound on the security of the Lunghi \etal protocol than the one derived here.

\newpage
\appendix

\begin{widetext}

\newpage

\section{Detailed proof of Theorem \ref{theorem}}

In this Appendix, we give a formal proof of Theorem \ref{theorem}. We consider the case of a cheating Alice. At round $j$, active Alice receives a string $b_j \in \F_q$ and sends back a message $y_j$. From the relativistic constraints, we know that this message $y_j$ is totally independent of $b_{j-1}$. We can therefore view $y_j$ as a function of $d,b_1,\dots,b_{j-2},b_j$. We also recursively define the functions $a_j = y_j + (b_j * a_{j-1})$, with $a_0 = d$. These are functions of $d,b_1,\dots,b_j$.

Note that if Alice's performs a probabilistic cheating strategy, her success probability will be the average of the success probabilities for each possible strategy she performs. It is therefore sufficient to bound Alice's cheating probability over all deterministic strategies. 
Let us then consider a deterministic cheating strategy for Alice: it is fully determined by the functions $y_j$, as well as a function $G(d,b_1,\dots,b_{k-1})$ that $\mathcal{A}_1$ uses to guess $a_k$ during the reveal phase. Alice successfully reveals $d$ iff
$[G(d,b_1,\dots,b_{k-1}) = a_k(d,b_1,\dots,b_k)]$. Therefore, we have 

\begin{align*}
1 + \eps & = \Pr[\mbox{Alice successfully reveals } d = 0] + 
\Pr[\mbox{Alice successfully reveals } d = 1] \\
& = \Pr_{b_1,\dots,b_k} [G(0,b_1,\dots,b_{k-1}) = a_k(0,b_1,\dots,b_k)]
+ \Pr_{b_1,\dots,b_k} G(1,b_1,\dots,b_{k-1}) = a_k(1,b_1,\dots,b_k)] \\ & = 2 \Pr_{d,b_1,\dots,b_k} [G(d,b_1,\dots,b_{k-1}) = a_k(d,b_1,\dots,b_k)].
\end{align*}

Intuitively, Alice will be able to win if the function $a_k$ is independent of $b_k$, on average on $d$ and the other $b_i$. We will prove that $a_k$ has some large dependence on $b_k$, which will limit Alice's cheating possibilities. We will actually show by induction that for each $j$, the function $a_j$ has some large dependency on $b_j$.

We define the independence parameter of function $f$ for a variable $y$ as follows :
\begin{definition}[Independence parameter of a variable on a function]
Let $f: \mathcal{X} \times \mathcal{Y} \to \mathcal{Z}$ be a function. The \emph{Independence Parameter} of $f$ for variable $y \in \mathcal{Y}$, denoted by $\IP{f}{y}$, is defined by

\begin{align}
\IP{f}{y} := \max_{g : \mathcal{X} \to \mathcal{Z} } \left[ \mathrm{Pr}_{x,y} \left[ f(x,y) = g(x)\right]\right],
\end{align} 
where we use the uniform measure on $\mathcal{X} \times \mathcal{Y}$.
\end{definition}

By definition, the case $\IP{f}{y} = 1$ corresponds to a function $f$ independent of $y$. If $\IP{f}{y} <1$, then the function $f$ depends on $y$. The definition of the independence parameter immediately yields $1 + \eps = 2 \Z{k}$, and our goal is therefore to obtain a tight upper bound for $\Z{k}$.

We prove the following :

\begin{proposition}
\label{Proposition:CHSHq}
$ \forall j, \ \IP{a_j}{b_j} \leq \frac{1}{2} + j \sqrt{\frac 2 q}.
$
\end{proposition}

\begin{proof}
We prove the proposition by induction on $j$. 

Let us first consider the base case:
\begin{align}
\Z{1} = \max_{g : \F_q \to \F_q} \Pr_{d,b_1} [a_1(d, b_1) = g(d)]
\end{align}
where $b_1$ is uniformly distributed in $\F_q$ and $d$ is equal to either $0$ or $1$, each with probability $1/2$. Let $g$ the function that maximizes the above expression, which gives $\Z{1} = \Pr_{d,b_1} [a_1(d,b_1) = g(d)]$. We write $a_1(d,b_1) = y_1(b_1) + (b_1 * d)$ for some function $y_1$. We now use the functions $g$ and $y_1$ to construct a strategy for a game $G \in \mathrm{CHSH}_q(1/2)$. We consider the following game between two players Adeline and Bastian :
\begin{itemize}
	\item Adeline receives a random element $X \in F_q$. Bastian receives an element $Y \in F_q$ which is equal to $0$ with probability $1/2$ and $1$ with probability $1/2$.
	\item Their goal is to respectively output $A$ and $B$ in $\F_q$ such that $A + B = X * Y$.
\end{itemize}
The above game is in $\mathrm{CHSH}_q(1/2)$. Intuitively, we mapped $\mathcal{A}_1$ to Adeline and $\mathcal{A}_2$ to Bastian, where the input $X$ corresponds to $b_1$ and the input $Y$ corresponds to $d$. 

We consider the following strategy for this game: Adeline outputs $A = y_1(X)$ and Bastian outputs $B = - g(Y)$. They win the game iff $y_1(X) - g(Y) = X * Y$. Therefore, we have
\begin{align*}
\omega(G) & \ge \Pr_{X,Y} [y_1(X) - g(Y) = X * Y] = \Pr_{X,Y} [a_1(Y,X) + (X * Y) - g(Y) = (X * Y)] \\ & = \Pr_{X,Y} [a_1(Y,X) = g(Y)] = \Z{1}.
\end{align*}
Combining this lower bound on the value $\omega(G)$ of the game with Lemma \ref{lemma} applied to $G \in \mathrm{CHSH}_q(1/2)$ gives $\Z{1} \le \omega(G) \leq \frac{1}{2} + \sqrt{\frac{2}{q}}$, which establishes the base case.

We now move to the induction step and assume that $\Z{j} \leq \frac{1}{2} + j \sqrt{\frac 2 q}$. Let us fix $h:=(d, b_1, \ldots, b_{j-1})$ the history before time $j$. Let us define the independence parameter conditioned on the history  $h$:
$$\Z{j+1}^h = \max_{g_{j+1} : \F_q \rightarrow \F_q} \Pr_{b_j,b_{j+1}}[a_{j+1}(h,b_j,b_{j+1}) = g_{j+1}(b_{j})].
$$
Averaging over $h$ gives back the independence parameter: $\Z{j+1} = \E_h [\Z{j+1}^h]$. 
We write $a_{j+1}(h,b_j,b_{j+1}) = y^h_{j+1}(b_{j+1}) + (b_{j+1} * a_{j}(h,b_j)) $. Notice that the dependence in $b_j$ of the function $a_{j+1}(h,b_j,b_{j+1})$ lies only in the function $a_{j}(h,b_j)$. Therefore, we can write
$$\Z{j+1}^h = \max_{{g}_{j+1} : \F_q \rightarrow \F_q} \Pr_{b_j,b_{j+1}}[a_{j+1}(h,b_j,b_{j+1}) = {g}_{j+1}(a_j(h,b_j))].
$$

 Let $g^h_{j+1}$ be the function that maximizes the expression:
$$\Z{j+1}^h = \Pr_{b_j,b_{j+1}}[a_{j+1}(h,b_j,b_{j+1}) = g^h_{j+1}(a_j(h,b_{j}))].
$$

We now use the functions $y^h_{j+1}$ and $g^h_{j+1}$ to construct a strategy for a game $G^h_{j+1} \in \mathrm{CHSH}_q(\Z{j}^h)$. We consider the following game between two players Adeline and Bastian : 
\begin{itemize}
	\item Adeline receives a random element $X \in F_q$. Bastian receives an element $Y \in F_q$ such that $\Pr[Y = c] = \Pr_{b_j}[a_j(h,b_j) = c]$.
	\item Their goal is to respectively output $A$ and $B$ in $\F_q$ such that $A + B = X * Y$
\end{itemize}

Intuitively, we mapped the active Alice (during round $j+1$) to Adeline and the passive Alice to Bastian, where the input $X$ corresponds to $b_{j+1}$ and the input $Y$ corresponds to $a_j$. Recall that the active Alice has no information about $b_j$ during step $j+1$. Therefore, she can determine $a_j$ with probability at most: 
$\Z{j}^h := \max_c \Pr_{b_j}[a_j(h,b_j) = c] $.
This shows that the above game $G^h_{j+1}$ is in $\mathrm{CHSH}_q(\Z{j}^h)$.

  We consider the following strategy for this game: Adeline outputs $A = y_{j+1}^h(X)$ and Bastian outputs $B = - g_{j+1}^h(Y)$. They win the game iff $y_{j+1}^h(X) - g_{j+1}^h(Y) = X * Y$, which implies that
\begin{align*}
\omega(G^h_{j+1}) & \ge \Pr_{X,Y} [y_{j+1}^h(X) - g^h_{j+1}(Y) = X * Y] \\
& = \Pr_{X,b_j} [y_{j+1}^h(X) - g^h_{j+1}(a_j(h,b_j)) = X * a_j(h,b_j)] 
\qquad  \mbox{where the distribution over both } X \mbox{ and } b_j \mbox{ is uniform} \\
&  = \Pr_{X,b_j} [a_{j+1}(h,b_j,X) + (a_j(h,b_j) * X) - g_{j+1}^h(a_j(b_j)) = (X * a_j(h,b_j))] \\
& = \Pr_{X,b_j} [a_{j+1}(h,b_j,X) = g_{j+1}^h(a_j(h,b_j))]\\
& = \Z{j+1}^h.
\end{align*}
Moreover, Lemma \ref{lemma} shows that  $\omega(G^h_{j+1}) \le \Z{j}^h + \sqrt{\frac{2}{q}}$ since the game $G$ belongs to $\mathrm{CHSH}_q(\Z{j}^h)$.
Combining both inequalities gives:
\begin{align}
\label{eqn}
\Z{j+1}^h \le \Z{j}^h + \sqrt{\frac{2}{q}}.
\end{align}

In order to conclude, notice that $\Z{j} = \E_h [\Z{j}^h]$ and $\Z{j+1} = \E_h [\Z{j+1}^h]$. Taking the expectation of Eq.~\ref{eqn} over the history $h$ finally gives:
$$ \Z{j+1} = \E_h [\Z{j+1}^h] \le \E_h \left[\Z{j}^h + \sqrt{\frac{2}{q}}\right] 
= \Z{j} + \sqrt{\frac{2}{q}} \le \frac{1}{2} + (j+1) \sqrt{\frac{2}{q}}. $$ 
\end{proof}
Proposition \ref{Proposition:CHSHq} implies that $\Z{k} = \frac{1}{2} + k \sqrt{\frac{2}{q}}$, and the discussion at the beginning of the appendix allows us to conclude that the protocol is $\eps$-binding with $\eps = 2 k \sqrt{\frac{2}{q}}$.
\end{widetext}

\end{document}